\documentclass[prl,twocolumn,english,superscriptaddress]{revtex4-1}

\usepackage{amsmath}
\usepackage{amssymb}
\usepackage{bbold}
\usepackage{graphicx}
\usepackage{color}
\usepackage[hang,flushmargin]{footmisc} 
\usepackage{amsthm}

\newcommand{\Hl}[0]{H_\text{loc}}
\newcommand{\ket}[1]{\left\vert#1\right\rangle}
\newcommand{\bra}[1]{\left\langle#1\right\vert}
\newcommand{\braket}[2]{\left\langle#1\vert#2\right\rangle}
\newcommand{\set}[1]{\left\lbrace#1\right\rbrace}
\newcommand{\comm}[2]{\left[#1,#2\right]}
\newcommand{\brak}[1]{\left\langle#1\right\rangle}

\newcommand{\tr}[0]{\text{Tr}}

\newcommand{\modulus}[1]{\left\vert#1\right\vert}

\theoremstyle{plain}
\newtheorem{theorem}{Theorem}
\newtheorem{lemma}[theorem]{Lemma}
\newcommand{\norm}[1]{\left\|#1\right\|}

\theoremstyle{definition}

\theoremstyle{remark}

\begin{document}
\renewcommand\abstractname{}

\title{Quantum Speed Limits for Quantum Information Processing Tasks}
\author{Jeffrey M. Epstein}
\affiliation{Berkeley Center for Quantum Information and Computation, Berkeley, California 94720 USA}
\affiliation{Department of Physics, University of California, Berkeley, California 94720 USA}

\author{K. Birgitta Whaley}
\affiliation{Berkeley Center for Quantum Information and Computation, Berkeley, California 94720 USA}
\affiliation{Department of Chemistry, University of California, Berkeley, California 94720 USA}
\date{\today}
\begin{abstract} 
We derive algebraic bounds on achievable rates for quantum state transfer and entanglement generation in general quantum systems. We apply these bounds to graph-based models of local quantum spin systems to obtain speed limits on these tasks. Comparison with numerical optimal control results for spin chains suggests that unexplored regions of the dynamical landscape may support enhanced performance of key quantum information processing tasks.
\end{abstract}

\maketitle

\newpage

\section{1. Introduction}
Robust and efficient quantum control is increasingly relevant to quantum science and technology. At present, the theory of quantum control of small systems (low-dimensional Hilbert spaces) is significantly more advanced than the corresponding theory for large systems (high-dimensional Hilbert spaces). While this situation reflects the current experimental state of the art, a complete toolkit for the quantum control of future experiments and devices must include strategies suited to both small and large systems.

In the low-dimensional setting, e.g. for one or two qubits, a Lie algebraic framework exists for finding optimal control protocols for many tasks \cite{D'Alessandro2007}. Unfortunately, the relevant algebraic tools and techniques are intractable in generic high-dimensional (many-body) cases. As a partial remedy to this obstacle, there is a large body of work on numerical techniques for finding efficient control protocols for quantum information processing in many-body systems. In the absence of analytic solutions to optimal control problems in this setting, bounding optimal times for accomplishing various tasks becomes important, see e.g. \cite{Ashhab2012}. Given the complexity of numerically-obtained optimal control sequences for simple tasks such as quantum state transfer, using such bounds to check the near-optimality of numerical solutions may be the only tractable option for high-dimensional quantum optimal control.

The problem of bounding optimal control times is also significantly different in the small and large system contexts. In small quantum systems, it is not unreasonable to suppose that a wide range of couplings is available, and ``quantum speed limits" such as those of Margolus-Levitin and Mandelstam-Tamm \cite{Mandelstam1991,Margolus1998} can be used to obtain meaningful bounds on the rate at which information processing tasks may be achieved. In Appendix A, we set out several such results for comparison with the many-body case studied here. It is important to note that such bounds reflect only the spectral properties of the Hamiltonian, and place no further constraint on its form.

In contrast to the low-dimensional case, many-body control problems must in general account for real-space properties of the system. While it may be possible to directly couple any two qubits in a small quantum processor, direct coupling of distant spins in a long spin chain may be impossible. For this reason, the bounds available in the low-dimensional setting tell us little about e.g. the minimum time required, given certain local interactions and external control fields, to apply a swap gate to a distant pair of qubits. In other words, the bare Hilbert space structure alone fails to capture some relevant information. Different techniques are required to obtain meaningful information.

Bravyi et al. showed \cite{Bravyi2006} that the Lieb-Robinson bound \cite{Lieb1972} can be used to obtain bounds on classical channel capacity and correlation formation in spin systems evolving under local Hamiltonians. In this paper, we use similar techniques to establish bounds on the speed with which high-fidelity quantum state transfer and entanglement generation can be performed in general quantum systems, then specialize to local spin systems. This bound is directly relevant to proposed quantum computer architectures based on spin chains, such as the nitrogen-vacancy center proposal of Yao et al. \cite{Yao2012}.

\section{2. Mathematical Framework}
Finite-dimensional quantum mechanics can be studied with no reference to spatial organization. In practice, however, some tensor product decompositions of Hilbert spaces have physical interpretations that are relevant for understanding what is experimentally achievable. An convenient way to formalize this notion of locality uses graphs to construct Hilbert spaces by associating small Hilbert spaces to each vertex. The full Hilbert space is taken to be the tensor product of the vertex spaces, so that any subset of vertices constitutes a subsystem. As in many other accounts, e.g. \cite{Nachtergaele2010}, we consider a graph $G=(V,E)$ with the following dictionary:
\begin{center}
\begin{tabular}{cccc}
Graph & & & Hilbert Space/Operator\\\hline
$v\in V$ & & & $\mathcal{H}_v$\\
$X\subseteq V$ & & & $\mathcal{H}_X=\bigotimes_{v\in X}\mathcal{H}_v$\\
$X\subset Y\subset V$  & & & $\mathcal{B}(\mathcal{H}_X)\simeq \mathcal{B}(\mathcal{H}_X)\otimes\mathbb{1}_{Y\setminus X}$
\end{tabular}
\end{center}
The isomorphism is the obvious one. For compactness of notation, elements of $\mathcal{B}(\mathcal{H}_X)$ will be called ``operators acting on $X$". $V$ is assumed to be finite, as are the $\mathcal{H}_v$.

A useful procedure \cite{Bravyi2006} that is natural to define in this graph picture is localization of an operator to a particular region, i.e. a subset of vertices. Given an operator $A$, the $X$-localization of $A$, $\left[A\right]_X$, is defined as
\begin{equation}
\left[A\right]_X=\int_{U(\bar{X})} UAU^\dagger d\mu(U),
\end{equation}
where $\mu$ is the Haar measure over the unitary group on $\bar{X}=V\setminus X$. Note that $\left[A\right]_X$ acts as the identity on $\bar{X}$ and $\norm{\left[A\right]_X}\leq \norm{A}$, where $\norm{\cdot}$ is the operator norm.

\section{3. Algebraic Bounds on Control}
In this section, we show that achievable rates of two important tasks in quantum information processing, quantum state transfer between subsystems and entanglement generation, obey bounds that follow directly from bounds on the norms of certain commutators. This allows the extensive work on such bounds (see e.g. \cite{Nachtergaele2010B}) to be used to obtain speed limits for quantum control tasks.

\subsection{Quantum State Transfer}

Suppose we would like to transfer a quantum state from region $X$ to region $Y$ of a local spin system (perhaps a spin chain) by applying an operator $T$, which may for instance be the time-evolution operator generated by some local Hamiltonian. To bound the speed with which this task can be accomplished, we must fix an appropriate figure of merit. One possible choice would be
\begin{equation}
\inf_{\rho}F\left(\tr_{\bar{X}}\rho,\tr_{\bar{Y}}\rho_T\right)
\end{equation}
where $\rho_T=T\rho T^\dagger$, $F$ is the fidelity, and the infimum is taken either over all density operators of the full system or perhaps over all density operators of the form $\rho_X\otimes\rho_{\bar{X}}$ for fixed $\rho_{\bar{X}}$. If this quantity is large, then $T$ can be used to transfer arbitrary states from $X$ to $Y$.

Unfortunately, upper bounding this figure of merit is difficult due to the presence of the infimum. A more convenient figure of merit follows from noting that if $T$ is able to effect state transfer for \textit{any} input state, then there must be some operator $A$ on $X$, which can be thought of as a state-preparation operator, such that 
\begin{equation}
F(\tr_{\bar{Y}}T\rho T^\dagger,\tr_{\bar{Y}}TA\rho A^\dagger T^\dagger)\label{figofmeritA}
\end{equation}
is small, since we must be able to transfer orthogonal pairs of states. In Appendix A, we show that this characterization of state transfer is related to the speed with which \textit{classical} information can be sent from one end of the chain to the other.

Now we can prove a bound on this figure of merit. Denote $\mathcal{O}_T=T\mathcal{O}T^\dagger$.
\begin{theorem}
Let $X$ and $Y$ be disjoint subsystems of a system $S$ in the initial state $\rho$. For some fixed unitary $T$ on $S$, suppose that for any $\mathcal{O}_X$, $\mathcal{O}_Y$ acting on $X$ and $Y$, respectively,
\begin{equation}
\norm{\comm{(\mathcal{O}_X)_T}{\mathcal{O}_Y}}\leq c_T(X,Y)\norm{\mathcal{O}_X}\norm{\mathcal{O}_Y}\label{cdef}
\end{equation}
holds, with $c_T$ a scalar function of  subystems $X$ and $Y$, for operator $T$. Then if $A$ is some operator on $X$, the fidelity between the reduced states of subsystem $Y$ given the overall states $T\rho T^\dagger$ and $TA\rho A^\dagger T^\dagger$ satisfies:
\begin{equation}
F\geq 1-c_T(X,Y)\norm{A}\label{Fbound}
\end{equation}
where we use the definition of the fidelity $F(\rho,\sigma)=\tr\sqrt{\sqrt{\sigma}\rho\sqrt{\sigma}}$ generalizing the pure state definition $F(\psi,\phi)=\modulus{\braket{\psi}{\phi}}$.
\end{theorem}

\begin{proof}
The trace distance between the two reduced states on system $Y$ is
\begin{equation}
\begin{aligned}
d&=\norm{\tr_{\bar{Y}}\left(A_T\rho_T A_T^\dagger-\rho_T\right)}_1.
\end{aligned}
\end{equation}
Using the triangle inequality, the properties of localized operators as defined above, and the monotonicity of the trace distance under partial trace, we obtain the following bound:
\begin{equation}
\begin{aligned}
d&\leq \norm{\tr_{\bar{Y}}\left(A_T\rho_T A_T^\dagger-\left[A_T\right]_{\bar{Y}}\rho_T \left[A_T^\dagger\right]_{\bar{Y}}\right)}_1\\
&\hspace{10pt}+\norm{\tr_{\bar{Y}}\left(\left[A_T\right]_{\bar{Y}}\rho_T \left[A_T^\dagger\right]_{\bar{Y}}-\rho_T\right)}_1\\
&\leq \norm{A_T\rho_T A_T^\dagger-\left[A_T\right]_{\bar{Y}}\rho_T \left[A_T^\dagger\right]_{\bar{Y}}}_1\\
&\leq 2\norm{A_T-\left[A_T\right]_{\bar{Y}}}.
\end{aligned}
\end{equation}
where the final inequality is proven in the lemma below. Note that the second term in the first line of this calculation vanishes because the operator $\left[A\right]_{\bar{Y}}$ acts as the identity on $Y$, so that when the partial trace is taken, the two resulting operators are the same. Following \cite{Bravyi2006}, we bound this norm distance by taking advantage of the unitary invariance of the Haar measure:
\begin{equation}
\begin{aligned}
\norm{A_T-\left[A_T\right]_{\bar{Y}}}&=\norm{A_T-\int UA_TU^\dagger d\mu(U)}\\
&=\int\norm{\comm{A_T}{U}} d\mu(U)\\
&\leq c_T(X,Y)\norm{A}
\end{aligned}
\end{equation}
where the integral is over the unitary group on $Y$. We conclude that $d\leq 2c_T(X,Y)\norm{A}$. Using the relation $F(\rho,\sigma)\geq 1-\frac{1}{2}\norm{\rho-\sigma}_1$ between the fidelity and the trace distance, we obtain the stated bound on fidelity, Eq. \ref{Fbound}.
\end{proof}

This theorem quantifies the relationship between the algebraic features of the operator $T$, as captured by the bound $c_T(X,Y)$ on the norms of the commutator Eq. \ref{cdef}, and its operational features. In particular, if $c_T(X,Y)$ is small, the influence of a local operator on system $X$ on the state of system $Y$ after application of $T$ is also small. The most interesting situations to consider are those in which local operations on $X$ and $Y$ may be applied at will, but $T$ is given, as might be the case for a pair of coupled qubits or a spin system with fixed interaction terms and variable local control fields.

Here we prove the lemma connecting the trace distance of density operators to the operator norm distance of unitary operators.
\begin{lemma}
	Let $A$ and $B$ be unitary operators on a finite-dimensional Hilbert space. Then for density operator $\rho$,
	\begin{equation}
	\norm{A\rho A^\dagger-B\rho B^\dagger}_1\leq 2\norm{A-B}.
	\end{equation}
\end{lemma}

\begin{proof}
	The norm difference on the left-hand side can be bounded by a supremum over operators on the Hilbert space as
	\begin{equation}
	\begin{aligned}
	\norm{A\rho A^\dagger-B\rho B^\dagger}_1&\leq\sup_{X\neq 0}\frac{\norm{AXA^\dagger-BXB^\dagger}_1}{\norm{X}_1}.
	\end{aligned}
	\end{equation}
	This is the trace norm $\norm{A\cdot A^\dagger-B\cdot B^\dagger}_1$ of the superoperator defined by $(A\cdot A^\dagger-B\cdot B^\dagger)(\rho)=A\rho A^\dagger-B\rho B^\dagger$. Using results from  \cite{Aharonov1998} for the superoperator trace and diamond norms for superoperators of this form, we obtain
	\begin{equation}
	\begin{aligned}
	\norm{A\cdot A^\dagger-B\cdot B^\dagger}_1&\leq \norm{A\cdot A^\dagger-B\cdot B^\dagger}_\diamond\\
	&\leq 2\norm{A-B}.
	\end{aligned}
	\end{equation}
\end{proof}

As a simple illustration of the main result of this section, consider a spin chain with $N$ sites evolving under a Hamiltonian $H$ that preserves the numbers of up and down spins. In other words, we have $\comm{H}{Z_1+\ldots+Z_N}=0$. Suppose that the chain is initialized in the state $\ket{\underline{1}}$ with all spins down except the first, then allowed to evolve under $H$ for time $t$. We would like to know the probability $p(t)$ of a spin flip at the $N^\text{th}$ site. Denoting spin down and up by $\ket{0}$ and $\ket{1}$, respectively, we have
\begin{equation}
\begin{aligned}
1-p(t)&=\bra{0}\left(e^{-iHt}\ket{\underline{1}}\bra{\underline{1}}e^{iHt}\right)_N\ket{0}\\
&=F^2\left(\left(e^{-iHt}\ket{\underline{1}}\bra{\underline{1}}e^{iHt}\right)_N,\left(e^{-iHt}\ket{\mathbf{0}}\bra{\mathbf{0}}e^{iHt}\right)_N\right)\\
&=F^2\left[\left(e^{-iHt}X_1\ket{\mathbf{0}}\bra{\mathbf{0}}X_1e^{iHt}\right)_N,\right.\\
&\hspace{60pt}\left.\left(e^{-iHt}\ket{\mathbf{0}}\bra{\mathbf{0}}e^{iHt}\right)_N\right]
\end{aligned}
\end{equation}
where $\ket{\mathbf{0}}$ is the state with all spins down. Applying the bound Eq. \ref{Fbound} and using the fact that $\norm{X_1}=1$ we find
\begin{equation}
1-p(t)\geq (1-c_t(\set{1},\set{N}))^2.
\end{equation}
where $c_t(\set{1},\set{N})$ is the Lieb-Robinson coefficient for regions separated by graph distance $N-1$ at times $t$ apart. For compactness, we have written $c_t$ instead of $c_{U(t)}$. See Section 4 for explicit expressions for this coefficient. Rearranging, we obtain
\begin{equation}
p(t)\leq c_t(N-1)\left[2-c_t(N-1)\right].
\end{equation}

\subsection{Entanglement Generation}

Another task of interest for quantum information processing is entanglement generation. Suppose that two distant regions begin in a separable state and we would like to entangle them by applying $T$. Can we do so? We start by proving a theorem showing that if all correlations between two subsystems, as measured by connected correlation functions of norm-bounded operators, are initially small and the constant $c_T(X,Y)$ (see Eq. \ref{cdef}) is bounded close to zero, the fidelity of the reduced state of $XY$ with any maximally entangled state after application of $T$ is bounded by a number less than one.
\begin{theorem}
	Let $\rho$ be a state of a bipartite $d\times d$-dimensional system $XY$ such that for any $A,B$ Hermitian operators on $X$ and $Y$, respectively, with $\norm{A},\norm{B}\leq 1$, the bound $\modulus{\brak{AB}_c}\leq f\leq 2/3$ on the magnitude of the connected correlator $\brak{AB}-\brak{A}\brak{B}$ holds. Then
	\begin{equation}
	F(\rho,\psi)\leq \sqrt{\frac{79}{81}+\frac{2f}{27}-\frac{f^2}{18}}\label{Sbound}
	\end{equation}
	for any maximally entangled state $\Psi$.
\end{theorem}
\begin{proof}
	Let $\ket{\Psi}$ be a maximally entangled state of $\mathbb{C}^d\otimes\mathbb{C}^d$. For an arbitrary density matrix $\rho$, let $\Delta=\rho-\ket{\Psi}\bra{\Psi}$. For $A\in\mathcal{B}(\mathbb{C}^d)\otimes\mathbb{1}_d$ and $B\in\mathbb{1}_d\otimes \mathcal{B}(\mathbb{C}^d)$ Hermitian with $\norm{A},\norm{B}\leq 1$, the connected correlation function of $A$ and $B$ in the state $\rho$ is
	\begin{equation}
	\begin{aligned}
	\brak{AB}_{c}&=\tr\left(\rho AB\right)-\tr\left(\rho A\right)\tr\left(\rho B\right)\\
	&=\bra{\Psi}AB\ket{\Psi}-\bra{\Psi}A\ket{\Psi}\bra{\Psi}B\ket{\Psi}\\
	&\hspace{10pt}+\tr\left(\Delta AB\right)-\tr\left(\Delta A\right)\tr\left(\Delta B\right).
	\end{aligned}
	\end{equation}
	Rearranging and taking the modulus:
	\begin{equation}
	\begin{aligned}
	&\modulus{\brak{AB}_{c}-\bra{\Psi}AB\ket{\Psi}+\bra{\Psi}A\ket{\Psi}\bra{\Psi}B\ket{\Psi}}\\
	&\hspace{20pt}\leq\modulus{\tr\left(\Delta AB\right)}+\modulus{\tr\left(\Delta A\right)}\modulus{\tr\left(\Delta B\right)}\\
	&\hspace{20pt}\leq \norm{\Delta}_1\norm{A}\norm{B}+\norm{\Delta}_1^2\norm{A}\norm{B}\\
	&\hspace{20pt}\leq \norm{\Delta}_1+\norm{\Delta}_1^2\leq 3\norm{\Delta}_1.
	\end{aligned}
	\end{equation}
	The last inequality used the fact that $\norm{\Delta}_1\leq\norm{\rho}_1+\norm{\ket{\Psi}\bra{\Psi}}_1=2$. Now using the reverse triangle inequality:
	\begin{equation}
	\begin{aligned}
	3\norm{\Delta}_1\geq \modulus{\modulus{\bra{\Psi}AB\ket{\Psi}-\bra{\Psi}A\ket{\Psi}\bra{\Psi}B\ket{\Psi}}-\modulus{\brak{AB}_{c}}}
	\end{aligned}.
	\end{equation}
	For any maximally entangled state, there are $A'$ and $B'$ such that $\brak{A'B'}_c\geq 2/3$, so that
	\begin{equation}
	\begin{aligned}
	3\norm{\Delta}_1\geq \modulus{\frac{2}{3}-\modulus{\brak{A'B'}_{c}}}
	\end{aligned}.
	\end{equation}
Now, it is given that for any $A$, $B$, $\modulus{\brak{AB}_c}\leq f$. Then
	\begin{equation}
	\begin{aligned}
	3\norm{\Delta}_1\geq\frac{2}{3}-f
	\end{aligned}
	\end{equation}
	where we can drop the modulus since by assumption $f\leq 2/3$. Then using $F^2\leq 1-\norm{\Delta}_1^2/2$ where $F$ is the fidelity of $\rho$ with $\ket{\Psi}\bra{\Psi}$, we obtain the stated bound Eq. \ref{Sbound}.
\end{proof}
Now we relate the correlation structure of the system after application of the operator $T$ to the algebraic structure of the time-evolved operators, as captured by Eq. \ref{cdef}:
\begin{theorem}\label{theorem_entanglement}
	Let a system $S$ be initialized in the state $\rho$ with the property that for any disjoint subsystems $X,Y\subset S$ and any $A$ and $B$ with $\norm{A},\norm{B}\leq 1$ acting on $X$ and $Y$, respectively, $\modulus{\brak{AB_c}}\leq f_0(X,Y)$. Then fixing subsystems $X$ and $Y$ and operators $A$ and $B$ on these, for any unitary operation $T$ on $S$ the following inequality holds for $\rho_T$:
	\begin{equation}
	\begin{aligned}
\modulus{\brak{AB}_c}&\leq f_0(Z,\bar{Z})\\
&\hspace{5pt}+2\left[\left(c_T(X,\bar{Z})+1\right)\left(c_T(Z,Y)+1\right)-1\right]
	\end{aligned}\label{corrbound}
	\end{equation}
	for any subsystem $Z$ such that $X\subseteq Z$ and $Y\subseteq \bar{Z}$.
\end{theorem}

\begin{proof}
Define $\Delta_A=A_T-\left[A_T\right]_{Z}$ and $\Delta_B=B_T-\left[B_T\right]_{\bar{Z}}$. Then we have
	\begin{equation}
	\begin{aligned}
	\modulus{\brak{A_TB_T}_c}&=\modulus{\brak{\left(\Delta_A+\left[A_T\right]_{Z}\right)\left(\Delta_B+\left[B_T\right]_{\bar{Z}}\right)}_c}\\
	&\leq\modulus{\brak{\left[A_T\right]_{Z}\left[B_T\right]_{\bar{Z}}}_c}+\modulus{\brak{\left[A_T\right]_{Z}\Delta_B}_c}\\
	&\hspace{10pt}+\modulus{\brak{\Delta_A\left[B_T\right]_{\bar{Z}}}_c}+\modulus{\brak{\Delta_A\Delta_B}_c}\\
	&\leq\modulus{\brak{\left[A_T\right]_{Z}\left[B_T\right]_{\bar{Z}}}_c}+2\norm{\Delta_B}\\
	&\hspace{10pt}+2\norm{\Delta_A}+2\norm{\Delta_A}\norm{\Delta_B}\\
	&\leq f_0(Z,\bar{Z})+2c_T(Y,Z)+2c_T(X,\bar{Z})\\
	&\hspace{10pt}+2c_T(X,\bar{Z})c_T(Y,Z)
	\end{aligned}
	\end{equation}
	where the last inequality was established in the proof of a previous theorem. Now we can switch into the Schr\"odinger picture, i.e., interpret this as a bound on $\modulus{\brak{AB}_c}$ in the state $\rho_T$. This is the stated bound, Eq. \ref{corrbound}.
\end{proof}
Note that the appearance of the subset $Z$ in Eq. \ref{corrbound} accounts for correlations between subregions of the spin system that may lead to correlations between regions $X$ and $Y$ after application of $T$.

Here again we have elucidated the relationship between algebraic and operational features of the operator $T$, this time to show that a small value of $c_T(X,Y)$ implies low fidelity of the state of the joint system $XY$ with any maximally entangled state of the two subsystems. This allows us to bound the entanglement-generating capabilities of $T$.

\section{4. Lieb-Robinson Bounds}
In this section, we restrict to the case in which $T$ is the time-evolution operator for a system built from a graph $G=(V,E)$ as described above, generated by a local Hamiltonian of the form
\begin{equation}
H(t)=\sum_{v\in V}\Phi_1(v,t)+\sum_{e\in E}\Phi_2(e,t)\label{LocH}
\end{equation}
where the graph-Hamiltonian dictionary is:

\begin{center}
	\begin{tabular}{cccc}
		Graph & & & Hamiltonian\\\hline
		$v\in V$ & & & $\Phi_1(v,t)\in\mathcal{B}(\mathcal{H}_v)$, $\norm{\Phi_1(v,t)}\leq B$\\
		$e\in E$ & & & $\Phi_2(e,t)\in\mathcal{B}(\mathcal{H}_e)$, $\norm{\Phi_2(e,t)}\leq J$
	\end{tabular}
\end{center}
We can now derive the Lieb-Robinson velocities for simple graphs that model cases of experimental importance. With these results, we shall convert the bounds Eq. \ref{Fbound} and Eq. \ref{Sbound} into concrete speed limits on quantum information processing tasks in local spin systems.

If the vertex spaces represent spin degrees of freedom, the $\Phi_1$ operators represent magnetic fields and the $\Phi_2$ operators nearest-neighbor spin-spin couplings. Many control problems assume that the $\Phi_2$ operators are time-independent and the $\Phi_1$ vary in time. This model captures, for example, the setting in which the couplings between spins are fixed but an experimentalist is free to vary some applied fields.

The Lieb-Robinson bound \cite{Lieb1972} demonstrates that a local Hamiltonian in the above sense implies a dynamical locality in the space of operators on the full graph Hilbert space. For our purposes a convenient statement is as follows:
\begin{theorem}
	Let $H$ be a local Hamiltonian of the form in Eq. \ref{LocH} for some graph $G=(V,E)$ with the correspondences described above. Then if $X,Y\subset V$ disjoint and $A$, $B$ are operators on $X$ and $Y$, respectively,
	\begin{equation}
	\begin{aligned}
	\norm{\comm{A(t)}{B}}&\leq 2\norm{A}\norm{B}\sum_{n=1}^\infty\frac{\left(2Jt\right)^n}{n!}N(n),
	\end{aligned}\label{LR}
	\end{equation}
	where $A(t)=U_t^\dagger AU_t$ for $U_t$ the time-evolution operator generated by $H(t)$ and $N(n)$ is the number of paths of length $n$ from $X$ to $Y$.
\end{theorem}
\noindent A more general form of this theorem was presented originally in \cite{Lieb1972}. An easier proof for time-independent Hamiltonian is given in \cite{Nachtergaele2010} and may be readily extended to time-dependent Hamiltonians (see Appendix B).

To find bounds on quantum state transfer and entanglement generation, we must find $c_{t}(X,Y)$ satisfying Eq. \ref{cdef}. This comes down to counting the number $N(n)$ of paths of length $n$ starting in $X$ and ending in $Y$. Suppose that $G$ is an arbitrary graph with maximum vertex degree $d$, and let $\text{dist}(X,Y)=R$, the minimum graph distance between vertices in subsets $X$ and $Y$. For $n<R$, $N(n)=0$. Otherwise, we have $N(n)\leq\modulus{X}d^n$. Then as in \cite{Richerme2014},
\begin{equation}
\begin{aligned}
c_{t}(X,Y)&\leq 2\sum_{n=R}^\infty\frac{\left(2Jt\right)^n}{n!}N(n)\\
&\leq 2 \modulus{X}e^{-R}\sum_{n=0}^\infty\frac{\left(2eJdt\right)^n}{n!}\\
&\leq 2\modulus{X}e^{2edJt-R}.
\label{Fbound_genG_vol}
\end{aligned}
\end{equation}

For $G$ a linear graph  ($d=2$), with $\text{diam}(X)<R$ and similarly for $Y$, notice that we can collapse all the vertices of $X$ (and any vertices surrounded by $X$) into a single vertex $v_X$ with associated Hilbert space $\mathcal{H}_{v_X}=\mathcal{H}_X$, and similarly for $Y$, to obtain a new graph $G'$ describing the same system but with the subsystems of interest now single vertices. This does not change the maximum strength of the edge interactions, nor does it change the degree of the graph, so we can assume without loss of generality that $\modulus{X}=\modulus{Y}=1$. Then we have $N(n)\leq C(n,\frac{1}{2}(n+R))$, so that
\begin{equation}
\begin{aligned}
c_{t}&\leq 2 \sum_{n=R}^\infty\frac{\left(2J t\right)^n}{n!}\binom{n}{\frac{n+R}{2}}\\
&=2I_{R}(4Jt)\label{Fbound_lin}
\end{aligned}
\end{equation}
where $I_\nu(x)$ is the modified Bessel function of the first kind.

The exponential form of the bound (last line of Eq. \ref{Fbound_genG_vol}) lends itself to interpretation as a speed limit $v=2edJ$ for arbitrary graph structure. The bound Eq. \ref{Fbound_lin} for the $d=2$ linear graph is not as convenient, but graphically (Fig. \ref{besselpic}) it can be seen to correspond to a speed limit with speed $6J$, an improvement over the $d=2$ case of the general limit in Eq. \ref{Fbound_genG_vol}.

\begin{figure}[h!]
	\centering
\includegraphics[width=.45\textwidth]{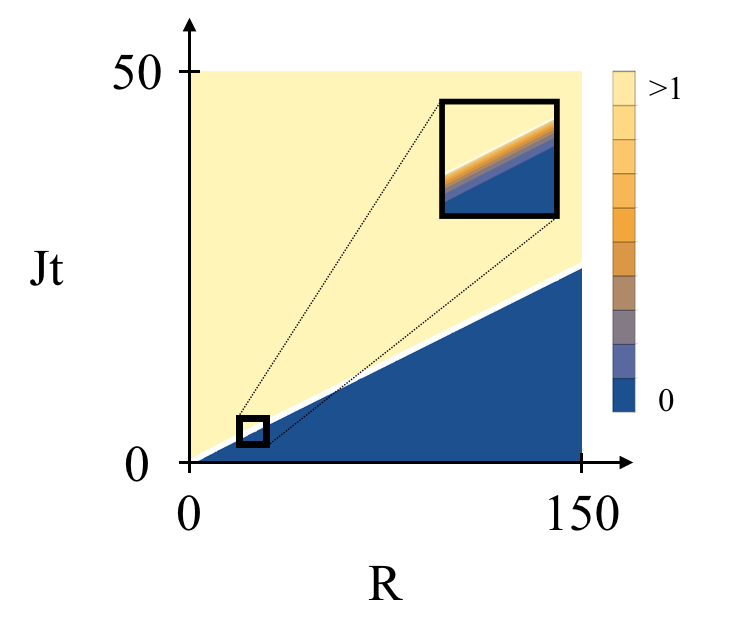}
\caption{Contour plot of $I_R(4Jt)$, the function controlling the bound on the fidelity of quantum state transfer in a spin chain (see Eq. \ref{Fbound}, Eq. \ref{Fbound_lin}). The function is truncated above 1, the bound $c\leq 2$ always holds by the triangle inequality.  A very sharp light cone with speed $6J$ emerges.}\label{besselpic}
\end{figure}

It is easy to insert these bounds into the bound Eq. \ref{Fbound} on quantum state transfer. We then see that the fidelity of quantum state transfer is exponentially suppressed outside a light cone defined by the Lieb-Robinson speed, i.e. for $R>6Jt$ in the case of the spin chain.

 To obtain an illuminating speed limit for entanglement generation, note that the time-dependent term in Eq. \ref{corrbound} is proportional to
\begin{equation}
\begin{aligned}
\left(c_T(X,\bar{Z})+1\right)\left(c_T(Y,Z)+1\right).
\end{aligned}
\end{equation}
If as above $c_T(X,Y)$ is a function of $R$ and $t$, i.e $c_T(X,Y)=g(R,t)$, let $Z$ be such that $\text{dist}(X,\bar{Z})=\text{dist}(Z,Y)=R/2$. Then this term can be bounded by
\begin{equation}
\begin{aligned}
\left(g(R/2,t)+1\right)^2.
\end{aligned}
\end{equation}
Recalling Eq. \ref{Fbound_genG_vol}, we see that the maximum speed associated to entanglement generation is twice that associated to quantum state transfer. This is reminiscent of the picture of entanglement spread via diverging Bell pairs proposed in \cite{Calabrese2005}.

\section*{5. Discussion}
As an example of the use of these bounds for determining the limits of control of quantum information tasks we consider the problem of quantum state transfer over a finite Heisenberg spin chain, i.e., a chain of length $L$ with nearest-neighbor couplings $-J\vec{\sigma}_n\cdot\vec{\sigma}_{n+1}/2$, under application of arbitrary local magnetic fields, i.e. 1-local control terms $B_n(t)\sigma_n^z$.  The two-body interaction terms have norm $J$, so application of Eq. \ref{Fbound_lin} yields a maximum transfer speed $v=6J$ in the presence of arbitrary magnetic fields, including time-dependent fields with local spatial variation along the chain.  In a previous study, Murphy et al. \cite{Murphy2010} searched numerically for optimal controls from a set of time-dependent magnetic fields, that would transfer a single spin state between the two ends of the chain. For chains of length $L$, they found control pulses that achieved high fidelity ($\leq \sim 10^{-4}$) quantum state transfer from one end of the chain to the other in times greater than $t_*\approx L/2J$ \cite{Murphy2010}.

The achievable fidelity was found to fall off rapidly for shorter transfer times, indicating a numerically extracted speed limit of $v\approx 2J$.  This is consistent with our result for quantum state transfer, which applies to a general system with \textit{any} 2-local interactions and 1-local control fields.  

It is interesting to consider the meaning of the gap between our bound $v=6J$ for general 2-local spin chain Hamiltonians and the numerically obtained maximal velocity. The latter was obtained for simulations restricted to the single-particle subspace, where the dynamics are amenable to treatment in terms of a group velocity.  Indeed, the value $v=2J$ is the maximum value of the group velocity for this system, which may be obtained from analytic solution of the Heisenberg Hamiltonian restricted to the single excitation subspace.

It has been previously noted that there can be a large gap between the group velocity, an intrinsic single-particle dynamical metric relevant to propagation of local excitations, and Lieb-Robinson bounds \cite{Richert2011, Richerme2014}. The latter hold for arbitrary excitations not restricted to single spins and thus are relevant also to more general quench dynamics, in particular to the global quenches that have been rationalized in terms of creation and subsequent interference of multiple entangled quasi-particle pairs \cite{Calabrese2006}.   Experiments with ion traps illustrate this distinction for local~\cite{Jurcevic2014} and global~\cite{Richerme2014} quenches of a finite chain of ions emulating the $XY$ model with approximately nearest neighbor interactions.  Specifically, Ref.~\cite{Jurcevic2014} indicates single excitation propagation velocities equal to the group velocity (Fig. 4e in~\cite{Jurcevic2014}) for interactions scaling as $1/r^{1.41}$,
while Ref.~\cite{Richerme2014} shows significantly higher velocities for propagation of correlations under global quenches (Fig. 3l in~\cite{Richerme2014}) although extraction of a velocity is problematic here since the system was not as well located in the nearest neighbor regime. For longer range power law ($1/r^{\alpha}$) interactions, recent theoretical work has shown that depending on the relative magnitude of the power law $\alpha$ and the lattice dimensionality $D$, generalizations of the Lieb-Robinson bounds can also allow finite information propagation velocities \cite{Hastings2006, eisert2013breakdown, gong2014persistence, foss2015nearly}. In this context it is interesting that Ref. \cite{eisert2013breakdown} also noted a striking gap between the generalized Lieb-Robinson bound and considerably smaller actual propagation times for a long range many-body Hamiltonian.

A second example of correlations propagating at a velocity greater than the single excitation group velocity following a global quench can be found in experiments with trapped atoms under conditions of restricted occupancy.  Although the Lieb-Robinson bounds do not apply in general to bosonic systems because the Hamiltonians are unbounded~\cite{Eisert2009}, a finite Hamiltonian norm is nevertheless obtained if the site occupancy is restricted to a fixed finite value.  The experiments in Ref.~\cite{Cheneau2012} fall into this category, restricting site occupancy to two or fewer atoms within an emulation of the Bose-Hubbard model in a finite one-dimensional chain of atoms trapped in an optical lattice. These experiments, and associated calculations in Ref.~\cite{Barmettler2012}, also show propagation velocities for correlation functions that are intermediate between group velocity and the Lieb-Robinson bound.

While qualitative, these recent experiments nevertheless indicate that there is a significant unexplored range of complex dynamics for increasingly efficient and fast quantum state transfer with non-local control fields.  It is thus an interesting challenge for engineering of spatio-temporal control fields to determine whether our commutator bound on quantum state transfer can be achieved.

Our results are related to those of Bravyi and co-workers in \cite{Bravyi2006}, but we have focused here more on constraints on the performance of quantum information processing tasks. For instance, whereas \cite{Bravyi2006} shows that a state obeying an area law for entanglement will evolve in finite time to another area law state under a local Hamiltonian, we examine the rate at which entanglement may form between two specific subsystems, not necessarily bipartitioning the entire graph. In principle, the bound of Bravyi et al. on the classical channel capacity between two subregions separated in space and time and linked by evolution under a local quantum Hamiltonian could be used to bound quantum state transfer times. Our formulation in terms of fidelities provides a direct and natural language for  general analysis of quantum information processing tasks.

In this work, we have used algebraic methods to prove the existence of fundamental limits on rates of quantum state transfer and entanglement generation. These limits could prove useful for understanding the ultimate limits of physical realizations of quantum information processing. Application to quantum spin systems with local interactions and control fields yielded Lieb-Robinson type bounds for quantum state transfer and entanglement generation. Comparison with the results of numerical optimal control calculations \cite{Murphy2010} for such systems suggests that unexplored regimes of quantum dynamics may offer new opportunities for enhanced performance of essential quantum processing tasks.\\
\\
\textit{Acknowledgments} We thank Pietro Silvi, Tommaso Calarco, Simone Montangero, and Daniel Burgarth for stimulating discussions. J.E. was supported by the Department of Defense (DoD) through the National Defense Science \& Engineering Graduate Fellowship (NDSEG) Program. This work was supported by Laboratory Directed Research and Development (LDRD) funding from Lawrence Berkeley National Laboratory, provided by the U.S. Department of Energy, Office of Science under Contract No. DE-AC02-05CH11231.

\bibliography{references}{}
\clearpage
\onecolumngrid
\appendix
\section{Appendix A: Orthogonalization Times and Quantum Information Processing}
A fundamental question in quantum mechanics is how fast a system can evolve from some initial state to an orthogonal one. This has been addressed before, for example in \cite{Mandelstam1991} and \cite{Margolus1998}. Here we use this approach to obtain bounds on the rates at which classical information, quantum information, and entanglement can be shared between the two parts of a bipartite system under generic Hamiltonian evolution (not necessarily local). Note that in general, these rates are much larger than those achievable with local Hamiltonians, as in the body of the paper.

We begin with a lemma that will allow us to construct explicitly the state that orthogonalizes the fastest under a given Hamiltonian evolution. This inequality then gives us a way to bound several interesting minimum times for any quantum system in a pure state.
\begin{lemma}
	Let $E_1\leq E_2\leq\cdots\leq E_N$ with $E_N-E_1\leq \pi$. Then for $M_{ij}=\cos(E_i-E_j)$, $r_i=(\delta_{i,1}+\delta_{i,N})/2$ minimizes $r^TMr$ subject to the constraints $r_i\geq 0$ and $\sum_ir_i=1$.
\end{lemma}
\begin{proof}
	Define the function $f(r)$ by
	\begin{equation}
	\begin{aligned}
	f(r)&=\sum_{ij}r_ir_jM_{ij}.
	\end{aligned}
	\end{equation}
	Then the first and second derivatives of $f$ are given by
	\begin{equation}
	\begin{aligned}
	\frac{\partial f}{\partial r_k}&=2\sum_jr_jM_{jk}\hspace{20pt}\frac{\partial^2f}{\partial r_j\partial. r_k}=2M_{jk}
	\end{aligned}
	\end{equation}
	From the value of $r$ in the statement of the lemma, we can move in the direction $\pm(\hat{r}_1-\hat{r}_N)$ and still satisfy the constraints. The first derivative in this direction is proportional to
	\begin{equation}
	\begin{aligned}
	\left(\frac{\partial}{\partial r_1}-\frac{\partial}{\partial r_N}\right)f&=2\sum_jr_j(M_{j1}-M_{jN})=M_{11}-M_{1N}+M_{N1}-M_{NN}=0
	\end{aligned}
	\end{equation}
	and the second derivative to
	\begin{equation}
	M_{11}-M_{1N}-M_{N1}+M_{NN}=2-M_{1N}-M_{N1}>0.
	\end{equation}
	We can also move in the direction $2\hat{r}_k-\hat{r}_1-\hat{r}_N$. In this direction, the first derivative is proportional to
	\begin{equation}
	\begin{aligned}
	2\frac{\partial}{\partial r_k}-\frac{\partial}{\partial r_1}-\frac{\partial}{\partial r_N}&=2\sum_jr_j\left(2M_{jk}-M_{j1}-M_{jN}\right)\\
	&=\left(2M_{1k}-M_{11}-M_{1N}+2M_{Nk}-M_{N1}-M_{NN}\right)=2\left(M_{1k}+M_{kN}-1-M_{1N}\right)\geq 0,
	\end{aligned}
	\end{equation}
	where we used that for $x,y,x+y\in\left[0,\pi\right]$,
	\begin{equation}
	\begin{aligned}
	\tan\left(\frac{y}{2}\right)&\leq\tan\left( \frac{\pi}{2}-\frac{x}{2}\right)=\cot\left(\frac{x}{2}\right)\\
	1&\geq\tan\left(\frac{x}{2}\right)\tan\left(\frac{y}{2}\right)=\left(\frac{\cos(x)-1}{\sin(x)}\right)\left(\frac{\cos(y)-1}{\sin(y)}\right)\\
	\sin(x)\sin(y)&\geq\left(\cos(x)-1\right)\left(\cos(y)-1\right)\\
	\cos(x)+\cos(y)-1&\geq\cos(x)\cos(y)-\sin(x)\sin(y)=\cos(x+y)
	\end{aligned}
	\end{equation}
	This establishes that $r_i=(\delta_{i,1}+\delta_{i,N})/2$ is a local minimum of the constrained optimization problem. Since all the $M_{jk}$ are non-negative, it is simple to verify that $f$ is convex on the region over which we're optimizing, where the $r_i$ are non-negative. Therefore the local minimum is a global minimum.
\end{proof}
\noindent\textbf{Whole System Orthogonalization:} Now we are in a position to find bounds on rates of orthogonalization. Let $\psi$ be a pure state of a $d$-dimensional system evolving under a Hamiltonian $H$ with energies $E_k$. Defining $\Delta_\text{max}=E_\text{max}-E_\text{min}$, we see that for times $t$ such that $\Delta_{max}t\leq \pi/2$,
\begin{equation}
\begin{aligned}
\modulus{\brak{\psi(t),\psi}}^2&=\sum_{j,k=0}^{d-1}r_jr_ke^{i(E_j-E_k)t}=\sum_{j,k}r_jr_k\cos(\Delta_{jk}t)\geq \frac{1}{2}\left(\cos(\Delta_\text{max}t)+1\right)=\cos^2\left(\frac{1}{2}\Delta_\text{max}t\right)
\end{aligned}
\end{equation}
with the eigenbasis basis chosen so that the $r_k$ are real. The inequality follows from the lemma proven above. Then the Bures angle between the initial and time $t$ states is bounded by \cite{Bosyk2014}
\begin{equation}
\begin{aligned}
\theta(\psi(t),\psi)&=\arccos\modulus{\brak{\psi(t),\psi}}\leq\frac{1}{2}\Delta_\text{max}t,
\end{aligned}
\end{equation}
resulting in an orthogonalization time $t=\pi/\Delta_\text{max}$. Note that the presence of $\Delta_\text{max}$ in the bound reflects the fact that entangled states are useful for estimation of parameters corresponding to classical fields (1-body operators). An entangled state takes advantage of the large $\Delta_\text{max}$ of the sum of many local operators, while a product state does not. In terms of fidelities, this is the difference between $(\cos\theta)^n$ and $\cos(n\theta)$.\\
\\
\textbf{Quantum Information Transfer:} Consider a bipartite system  in a pure state $\psi$ with $\tr_1\left(\psi\psi^\dagger\right)=\rho^{(2)}$. Then since the Bures angle is non-decreasing under partial trace, we also have
\begin{equation}
\begin{aligned}
\theta(\rho^{(2)}(t),\rho^{(2)})&\leq\frac{1}{2}\Delta_\text{max}t.
\end{aligned}
\end{equation}
We define $t_*^Q=\pi/\Delta_\text{max}$. This is the minimum time required after an operation on subsystem 1 for subsystem 2 to evolve to an orthogonal state. Since in order to send quantum information, the sender should be able to cause the receiver's system to evolve to any state, this is a reasonable measure of the minimum time to send a qubit.\\
\\
\textbf{Classical Information Transfer:} Now let $\psi_A(t)=e^{-iHt}A\psi$ and $\psi_B(t)=e^{-iHt}B\psi$ for some $A,B\in\mathcal{B}(\mathcal{H})$. By the triangle inequality:
\begin{equation}
\begin{aligned}
\theta(\psi_A(t),\psi_B(t))&\leq \theta(\psi_A(t),A\psi)+\theta(A\psi,B\psi)+\theta(B\psi,\psi_B(t))\\
&\leq \theta(A\psi,B\psi)+\Delta_\text{max}t.
\end{aligned}
\end{equation}
If $\theta(A\psi,B\psi)=\pi/2$, then this bound is trivial. However, consider the situation in which the system is bipartite and $A,B$ act as the identity on subsystem 2. Then
\begin{equation}
\begin{aligned}
\theta_B(\rho^{(2)}_A(t),\rho^{(2)}_B(t))&\leq\theta_B(\rho^{(2)}_A(t),\rho^{(2)}_A)+\theta_B(\rho^{(2)}_A,\rho^{(2)}_B)+\theta_B(\rho^{(2)}_B,\rho^{(2)}_B(t))\\
&\leq\theta_B(\psi_A(t),\psi_A)+\theta_B(\rho^{(2)}_A,\rho^{(2)}_B)+\theta_B(\psi_B,\psi_B(t))\\
&\leq\theta_B(\rho^{(2)}_A,\rho^{(2)}_B)+\Delta_{\text{max}}t.
\end{aligned}
\end{equation}
Using that $\theta_B(\rho^{(2)}_A,\rho^{(2)}_B)=0$, we find
\begin{equation}
\begin{aligned}
\theta_B(\rho^{(2)}_A(t),\rho^{(2)}_B(t))&\leq\Delta_{\text{max}}t.
\end{aligned}
\end{equation}
We define $t_*^C=\pi/2\Delta_\text{max}$. Since this is the minimum time required for the reduced states on subsystem 2 conditioned on the choice of one of two operations on subsystem 1 (a classical random variable) to become perfectly distinguishable, this is a reasonable measure of the minimum time to send a classical bit.\\
\\
\textbf{Entanglement Generation:} Let $\psi$ be a product state of a bipartite $d^2$-dimensional system and let $\phi$ be maximally entangled. Using the Schmidt basis for $\phi$:
\begin{equation}
\begin{aligned}
\modulus{\brak{\phi,\psi}}&=\modulus{\left(\frac{1}{\sqrt{d}}\sum_k\bra{kk}\right)\left(\sum_{ij}\alpha_i\beta_j\ket{ij}\right)}=\frac{1}{\sqrt{d}}\modulus{\sum_{k}\alpha_k\beta_k}\leq \frac{1}{\sqrt{d}}.
\end{aligned}
\end{equation}
Then $\theta_B(\phi,\psi)\geq\arccos d^{-1/2}$ so that the minimum time to generate a maximally entangled state from a product state is:
\begin{equation}
t_*^E=\frac{2}{\Delta_\text{max}}\arccos d^{-1/2}.
\end{equation}

These bounds can all be achieved, as may be demonstrated constructively using e.g. a system of two qubits evolving under the Hamiltonian $H=Z_1Z_2$ from the initial state $\ket{0+}$ with operators $A=X_1$, $B=\mathbb{1}$ (for the information transfer times) or initial state $\ket{++}$ (for entanglement time). Thus we have
\begin{equation}
t_*^C\leq t_*^E<t_*^Q=2t_*^C.
\end{equation}

In this sense, classical information can be sent from one part of a bipartite system to another twice as fast as quantum information. This leads to an interpretation of teleportation as a way to beat a quantum speed limit using entanglement as a resource if we imagine using the quantum dynamics of the coupled systems to do the necessary classical communication. This interpretation is bolstered by the fact that first generating entanglement and then performing teleportation takes at least as long as the minimum quantum transfer time.

It is not clear from this analysis that the distinction between transfer times for classical and quantum information would persist in the situation where the interaction between the sender and the receiver is mediated by intervening subsystems, as is the case in the spin chain. In the two-qubit example give above, the halving of the transfer time is directly related to the existence of what might be termed a ``local time-reversal operator", i.e. an operator acting non-trivially only on subsystem 1 that anti-commutes with the Hamiltonian. Such an operator cannot exist in the case of the spin chain with local interactions.

\section{Appendix B: Lieb-Robinson Bounds with Time-Dependent Hamiltonians}
A very clear proof of the Lieb-Robinson bound, in the form given in Eq. \ref{LR}, may be found in \cite{Nachtergaele2010} for the case of time-independent Hamiltonians. For application to quantum control, we need to extend the result to the time-dependent setting. This may be accomplished by a minor modification of the proof given in that work. We present only the modified step here.\\
\\
Eq. 2.27 of \cite{Nachtergaele2010} defines the function
\begin{equation}
f(t)=\comm{T_t(A)}{B}=\comm{\tau_t\left(\tau_{-t}^\text{loc}(A)\right)}{B}
\end{equation}
where $A$ and $B$ are operators on regions $X$ and $Y$ of a local spin system, $\tau_t$ is the time-evolution superoperator corresponding to a local Hamiltonian $H$, and $\tau_t^\text{loc}$ is the time-evolution superoperator corresponding to the Hamiltonian obtained by getting rid of all terms in $H$ that couple $X$ and its complement $\bar{X}$. In particular, note that $\tau_t^{\text{loc}}$ has $\mathcal{B}(X)$ as an invariant subalgebra.\\
\\
To extend to the time-dependent case, we replace $\tau_t$ by $\tau_{t\rightarrow t_0}$, as the superoperator is no longer time-invariant. Now we may compute the time-derivative of the first term in the commutator. $U(t\leftarrow s)$ is the unitary operator corresponding to time-evolution from time $s$ to time $t$, $\Phi(Z,t)$ is the operator in the Hamiltonian at time $t$ acting on subset $Z$ of vertices, and $S_{\Lambda(X)}$ is the set of all subsets of vertices with non-empty intersection with both $X$ and $\bar{X}$.
\begin{equation}
\begin{aligned}
\frac{d}{dt}\tau_{t\leftarrow 0}\left(\tau^\text{loc}_{0\leftarrow t}\left(A\right)\right)&=\frac{d}{dt}\left[U^\dagger(t\leftarrow 0)U^\dagger_\text{loc}(0\leftarrow t)AU_\text{loc}(t \rightarrow 0)U(0 \rightarrow t)\right]\\
\\
&=iU^\dagger(t\leftarrow 0)H(t)U^\dagger_\text{loc}(0\leftarrow t)AU_\text{loc}(t \rightarrow 0)U(0 \rightarrow t)\\
&\hspace{20pt}-iU^\dagger(t\leftarrow 0)H_\text{loc}(t)U^\dagger_\text{loc}(0\leftarrow t)AU_\text{loc}(t \rightarrow 0)U(0 \rightarrow t)\\
&\hspace{20pt}+iU^\dagger(t\leftarrow 0)U^\dagger_\text{loc}(0\leftarrow t)AU_\text{loc}(t \rightarrow 0)H_\text{loc}(t)U(0 \rightarrow t)\\
&\hspace{20pt}-iU^\dagger(t\leftarrow 0)U^\dagger_\text{loc}(0\leftarrow t)AU_\text{loc}(t \rightarrow 0)H(t)U(0 \rightarrow t)\\
\\
&=i\tau_{t\leftarrow 0}\left(H(t)\right)T_t(A)-i\tau_{t\leftarrow 0}\left(\Hl(t)\right)T_t(A)+iT_t(A)\tau_{t\leftarrow 0}\left(\Hl(t)\right)-iT_t(A)\tau_{t\leftarrow 0}\left(H(t)\right)\\
\\
&=i\comm{\tau_{t\leftarrow 0}\left(H(t)\right)}{T_t(A)}-i\comm{\tau_{t\leftarrow 0}\left(\Hl(t)\right)}{T_t(A)}\\
\\
&=i\comm{\tau_{t\leftarrow 0}\left(H(t)-\Hl(t)\right)}{T_t(A)}\\
\\
&=i\sum_{Z\subset S_\Lambda(X)}\comm{\tau_{t\leftarrow 0}\left(\Phi(Z,t)\right)}{T_t(A)}.
\end{aligned}
\end{equation}
The time derivatives of the unitary time-evolution operators follow from the time-dependent Schr\"odinger equation. Now the time derivative of the function $f(t)$ is
\begin{equation}
\begin{aligned}
\frac{d}{dt}f(t)&=\comm{\frac{d}{dt}T_t(A)}{B}\\
\\
&=i\sum_{Z\subset S_\Lambda(X)}\comm{\comm{\tau_{t\leftarrow 0}\left(\Phi(Z,t)\right)}{T_t(A)}}{B}\\
\\
&=-i\sum_{Z\subset S_\Lambda(X)}\comm{\comm{T_t(A)}{B}}{\tau_{t\leftarrow 0}\left(\Phi(Z,t)\right)}-i\sum_{Z\subset S_\Lambda(X)}\comm{\comm{B}{\tau_{t\leftarrow 0}\left(\Phi(Z,t)\right)}}{T_t(A)}\\
\\
&=i\sum_{Z\subset S_\Lambda(X)}\comm{\tau_{t\leftarrow 0}\left(\Phi(Z,t)\right)}{f(t)}-i\sum_{Z\subset S_\Lambda(X)}\comm{T_t(A)}{\comm{\tau_{t\leftarrow 0}\left(\Phi(Z,t)\right)}{B}}.
\end{aligned}
\end{equation}
This matches Eq. 2.28 of \cite{Nachtergaele2010}, and the rest of the proof proceeds as in that work.

\end{document}